\documentclass{amsart}
\usepackage{amsmath,amsthm}
\usepackage{amsfonts,amssymb}
\usepackage{amsfonts}
\usepackage[all]{xy}
\usepackage{pdfsync}
\usepackage{hyperref}
\usepackage{graphicx}
\usepackage{subfig}

\usepackage{tikz}%
\usepackage{color}
\usepackage{booktabs}
 \usepackage{multirow}
\usepackage{caption}
\usepackage{mathrsfs}
\usepackage{braket}
\usepackage{float}%提供float浮动环境
\allowdisplaybreaks[4]
\newtheorem{theorem}{Theorem}[section]

\theoremstyle{definition}
\newtheorem{definition}[theorem]{Definition}
\newtheorem{ex}[theorem]{Example}

\theoremstyle{remark}
\newtheorem{remark}[theorem]{Remark}
\numberwithin{equation}{section}

\begin{document}

\title[Incompatibility of  rank-one POVMs and quantum uncertainty relations]{Incompatibility of  rank-one POVMs and quantum uncertainty relations}

   %Only \author and \address are required; other information is
   %optional.  Remove any unused author tags.
%
  % author one information
%\author[short version for running head]{name for top of paper}
\author[X. Wang]{Xu Wang}
\address[1]{School of Mathematical Sciences\\ Tianjin Normal University\\
Binshui West Road 393, 300387 Tianjin\\ P.R. China.}
\email{16600368292@163.com}

\author[W. Dong]{Weisong Dong}
\address[2]{School of Mathematics\\ Tianjin University\\ 300354 Tianjin\\ P.R. China}
\email{dr.dong@tju.edu.cn}
\thanks{}

\thanks{}

   %author two information

\author[P. Lian]{Pan  Lian}
\address[3]{School of Mathematical Sciences\\ Tianjin Normal University\\
Binshui West Road 393, 300387 Tianjin\\ P.R. China.
}
%\curraddr{}
%\email{panlian@tjnu.edu.cn}
%\thanks{}

%\author[F. Maes]{Frederick Maes}
%\address{Research Group NaM2\\Department of Electronics and Information Systems\\Faculty of Engineering and Architecture\\ Ghent University\\
%Krijgslaan 281, 9000 Ghent\\ Belgium}
%\curraddr{}
%\email{Frederick.Maes@UGent.be}
%\thanks{}

%    \subjclass is required.
%\subjclass[2020]{Primary 30G35; Secondary  33C52, 33C55}
%\keywords{Incompatibility, POVM, Uncertainty relation, support, spark.}
\date{}

%    "Communicated by" -- provide editor's name; required.
%\commby{}

%    Abstract is required.
\begin{abstract} The incompatibility of quantum measurements  is a fundamental feature of quantum mechanics with profound implications for uncertainty relations and quantum information processing. In this paper, we extend the notion of {\em $s$-order incompatibility} of measurements, introduced  by Xu (inspired by De Bi\`evre's ``complete incompatibility" ), to more general rank-one POVMs, and establish novel uncertainty relations. Furthermore, we investigate the incompatibility of multiple POVMs and its connection to support uncertainty relations. These results  may have applications in quantum cryptography, quantum state reconstruction, and quantum compressed sensing.
\end{abstract}

\maketitle
\section{Introduction}
%%%%%%%%%%%%%%%%%%%%%%%%%%%%%%%%%%%%%%%%%%%%%%%%%%%%%%%%
The incompatibility of measurements  is a fundamental feature of quantum mechanics, setting it apart from classical physics.  Recent efforts to rigorously quantify incompatibility have revealed  profound connections to uncertainty relations, quantum state reconstruction, etc.

 Traditionally, incompatibility has been studied using  the  commutator  of quantum operators  \cite{ghk}. Recent studies suggest that  the alternative operator, i.e.  anticommutator or Jordan product  as an equally important  descriptor of measurement incompatibility. This  can be find some clue from the  mathematical fact:  for any non-negative rank-one operators $A$ and $B$, the non-negativity condition $\{A, B\} \geq 0$ holds if and only if they commute, i.e., $[A, B] = 0$. This equivalence highlights a deep interplay between algebraic noncommutativity and geometric negativity in operator space, revealing new perspectives on quantum coherence and measurement incompatibility, see e.g., \cite{gps, GL}.
\vspace{3pt}

A different approach to characterizing incompatibility was introduced by De Bi\`evre in \cite{deb}, who introduced the concept of {\em complete incompatibility} for two orthonormal bases of a Hilbert space $\mathscr{H}$,  linking it to the support uncertainty inequality and the identification of non-classical quantum states whose Kirkwood–Dirac quasiprobability distributions fail to be proper probability distributions. De Bi\`evre's notion was further extended by Xu in \cite{xu} to the {\em $s$-order incompatibility}, which provides a discrete measure (taking only integers numbers) of incompatibility rather than a continuous one, making it significantly different from previous formulations, such as \cite{GL}. Despite these advancements, an interesting question remains: how can we generalize these notions beyond orthonormal bases to more general measurement frameworks such as Positive Operator-Valued Measures (POVMs)?

\vspace{5pt}
We  observe  that the {\em complete incompatibility} or the more general {\em $s$-order incompatibility} for two orthonormal bases can be characterized using a concept from linear algebra — the spark of the matrix formed by concatenating the basis vectors. In the context of compressed sensing, the spark represents the smallest number of linearly dependent columns in a matrix, providing a concrete and computable measure of incompatibility. This insight offers a novel linear algebraic perspective on measurement incompatibility, linking it to sparsity and compressed sensing techniques.\vspace{5pt}

The main contribution of this paper is  extending the {\em $s$-order incompatibility} to general rank-one positive operator-valued measures (POVMs) via  frame theory \cite{ef}, we address the challenge of characterizing non-projective measurement incompatibility.  This extension reveals that the incompatibility of POVMs is governed by the linear dependence structure of their measurement vectors, making it particularly useful where non-projective measurements are essential, such as quantum state estimation with noise resilience. Building on this framework, we establish a new class of uncertainty relations in Section \ref{se4} that emphasize the interplay between measurement incompatibility and the minimal support of quantum states across different measurement settings,  in the sprit of the Ghobber-Jaming uncertainty relation \cite{gj}. In Section \ref{se5}, the original Ghobber-Jaming uncertainty is generalized in a different way using coherence  to rank-one POVMs, addressing a question raised in \cite{kk} for general frames.  

Furthermore, we extend our framework to systems involving multiple rank-one POVMs, introducing a generalized {\em $s$-order incompatibility} criterion. Recall that variance-based uncertainty relations have been studied for multiple observables associated with different bases, see e.g., \cite{huang, kw}. Our approach leads to  multi-measurement support uncertainty relation, which will be useful in applications requiring compatibility constraints across several observables, such as quantum error correction and multi-parameter metrology. Moreover, our support uncertainty relation has sharp lower bound, in contrast to the continuous setting, where establishing sharp bounds for multiple measurements remains  open. 

These findings not only deepen our understanding of measurement incompatibility but also have potential applications in quantum cryptography, quantum state reconstruction, and quantum error correction. The spark-based characterization may provide new insights into designing and analyzing measurement schemes in quantum compressed sensing and quantum tomography.
\vspace{5pt}

{\em Paper Organization and Notations}: The rest of this paper is structured as follows. Section 2 establishes the equivalence between {\em $s$-order incompatibility} and the spark of the measurement matrix. Section 3 extends the $s$-order incompatibility to rank-one POVMs via frame theory. Section 4 derives novel uncertainty relations within the incompatibility framework. Section 5 provides an alternative generalization of Ghobber-Jaming uncertainty relations for frames. Section 6 investigate the implications of incompatibility in multi-basis settings and its connection to support uncertainty. We conclude this paper by discussing potential applications and open questions. 

Throughout the paper, the set of consecutive integers $\{1, 2, \ldots, d\}$ is denoted by $[[1,d]]$.
For two orthogonal bases $A$ and $B$, we denote $(A, B)$ by the matrix whose columns are given by the basis of $A$ and $B$. For a set $S$, the notion $|S|$ denotes the number of elements in $S$, ${\rm span}\{S\}$ denotes the space spanned by $S$ over the complex field $\mathbb{C}$. Let $I=\{1,2,\ldots, m\}$ and $J=\{1,2,\ldots, n\}$.
For a set $E\subset I$, we will write $E^{c}$ for its complement.

\section{$s$-order incompatibility and spark}

Let $A=\{\ket{a_{j}}\}_{j=1}^{d}$ and $B=\{\ket{b_{k}}\}_{k=1}^{d}$ be two orthonormal bases for a $d$-dimensional  complex Hilbert space $\mathscr{H}$. To eliminate the phase ambiguity $\ket{a_{j}}\rightarrow e^{i\theta_{j}}\ket{a_{j}}$ with $\theta_{j}\in \mathbb{R}$ real number and $i=\sqrt{-1}$,  we consider as usual the sets of rank-one projectors (for a precise definition,  see Definition \ref{ran1} below)
\begin{equation*}
    \overline{A}=\{\ket{a_{j}}\bra{a_{j}}\}_{j=1}^{d} \qquad{\rm and} \qquad \overline{B}=\{\ket{b_{k}}\bra{b_{k}}\}_{k=1}^{d}. 
\end{equation*}
The bases $A$ and $B$ are said to be compatible if the projectors  $\overline{A}$ and $\overline{B}$ commute; otherwise, they are considered incompatible. It is worthwhile to further classify the degrees of incompatibility. De Bi\`{e}vre introduced the purely algebraic notion  {\em complete incompatibility} in \cite{deb},  which  is further refined  by Xu as follows.

\begin{definition}[$s$-order incompatibility \cite{xu}] \label{def1} Two orthonormal bases  $A$ and $B$ are said to be  $s$-order incompatible, where $s\in [[2, d+1]]$, if 
\begin{enumerate}
    \item For all nonempty subsets $S_{A}\subset A$ and $S_{B}\subset B$ with $|S_{A}|+|S_{B}|<s$, it holds that \[{\rm span}\{S_{A}\} \cap {\rm span}\{S_{B}\}=\{0\}.\]
    \item There exist subsets $S_{A}$ and $S_{B}$ with $|S_{A}|+|S_{B}|=s$, such that  \[{\rm span}\{S_{A}\}\cap {\rm span}\{S_{B}\}\neq \{0\}.\]
\end{enumerate}  
\end{definition}
\begin{remark}
   This concept provides a characterization for the degree of incompatibility. 
\end{remark}

    Alternatively,  the concept of  a matrix's spark  is widely used in compressed sensing, see e.g., \cite{de}.
\begin{definition}[Spark or Kruskal rank] The  spark of a $m\times n$ matrix $M$ is the smallest integer $k$ such that there exists a set of $k$ columns in $M$ that are linearly dependent.
\end{definition}

It is interesting to observe the following connection.

\begin{theorem}  Two orthonormal bases  $A$ and $B$ are $s$-order incompatible if and only if  ${\rm spark}(A, B)=s,$
where $(A, B)$ is the matrix whose column vectors are the basis vectors.
\end{theorem}

\section{Incompatibility of  tight frames and  rank-one POVMs}

In this section, we classify the incompatibility of more general quantum measurements. Recall that a Positive Operator-Valued Measure (POVM)  is a set of non-negative Hermitian operators $\{Q_{j}\}$, which are not necessarily projectors,  satisfying $\sum_{j}Q_{j}={\rm \mathbb{I}}$, where $\rm \mathbb{I}$ is the identity operator on $\mathscr{H}$. We focus on the following specific class of POVMs.
\begin{definition}\label{ran1} A quantum measurement is called  rank-one if its  measurement operators are rank-one, i.e.,  they take the outer-product form $Q_{i}=\mu \mu^{*}$ for some nonzero vectors $\mu\in \mathscr{H}$, where $\mu^{*}$ denotes the conjugate transpose of $\mu$. These associated vectors  are referred to as  measurement vectors.
 \end{definition}
\begin{remark}
     A rank-one POVM  generalizes standard quantum measurement, as its measurement vectors $\{\mu_{j}\}$ need not to be normalized or orthogonal.
\end{remark}
  There exists a fundamental connection between rank-one POVMs and the following so-called tight frames \cite{ef}. %Frames generalize bases and were introduced in the context of nonharmonic Fourier series by Duffin and Schaeffer in \cite{ds}. 

\begin{definition} Let $\{\varphi_{j}, 1\le j\le n\}$
denote a set of (overcomplete) vectors in a $d$-dimensional space $\mathscr{H}$. The vectors $\{\varphi_{j}\}$  is said to form a tight frame if there exists a constant $\alpha>0$
such that for all $x\in \mathscr{H}$,
\begin{equation*}
    \sum_{j=1}^{n}|\langle x, \varphi_{j}\rangle|^{2}=\alpha^{2}\|x\|^{2}.
\end{equation*}
 If $\alpha=1$, the tight frame is said to be normalized, otherwise, it is  $\alpha$-scaled.
\end{definition}
More precisely, the connection between rank-one POVMs and tight frames is connected as below.
\begin{theorem}\cite{ef} A set of vectors $\varphi_{j}\in  \mathscr{H}$ forms a $\alpha$-scaled tight frame for $\mathscr{H}$ if and only if the scaled vectors $\alpha^{-1}\varphi_{j}$ are the measurement vectors of a rank-one POVM on $\mathscr{H}$. In particular, the vectors $\varphi_{j}$ form a normalized tight frame for $\mathscr{H}$ if and only if they are the measurement vectors of a rank-one POVM on $\mathscr{H}$.
\end{theorem}
\begin{remark}
    For a  pure state $\ket{\varphi}$, the probability of observing the $i$th outcome is 
\begin{equation*}
    p(i)=\langle \varphi, Q_{i}\varphi\rangle = |\langle \mu_{i}, \varphi \rangle|^{2}.
\end{equation*} 
obviously, the probabilities $p(i)$ sum to $1$.
\end{remark}

%When there is no risk of confusion, we will slightly abuse the notation use the same notation for both %the POVMs and the corresponding set of  measurement vectors.

Using this connection,  we classify the incompatibility of two rank-one POVMs in terms of tight frames.
\begin{definition}\label{ed1} Suppose $A=\{a_{k}, 1\le k\le m\}$ and $B=\{b_{j}, 1\le j\le n\}$ are two tight frames in $\mathscr{H}$, with frame constant $\alpha$ and $\beta$, respectively.  If an integer $s$ satisfies the following conditions:
\begin{enumerate}
    \item  For any non-empty subsets $S\subseteq I$ and $J\subseteq J$ with $|S|+|T|<s$, and for any non-zero $x\in \mathscr{H}$, at most one of  the following holds,
    \begin{equation*}
        \sum_{k\in S}|\langle x, a_{k}\rangle|^{2}=\alpha\|x\|^{2} \quad {\rm and} \quad \sum_{j\in T}|\langle x, b_{j}\rangle|^{2}=\beta\|x\|^{2}.
    \end{equation*}  
    \item There exist non-empty subsets $S\subseteq I$ and $T\subseteq J$ with $|S|+|T|=s$, and a  a nonzero  $x\in \mathscr{H}$, such that
    \begin{equation*}
        \sum_{k\in S}|\langle x, a_{k}\rangle|^{2}=\alpha\|x\|^{2} \quad {\rm and} \quad \sum_{j\in T}|\langle x, b_{j}\rangle|^{2}=\beta\|x\|^{2}.
    \end{equation*}
\end{enumerate}
Then the frames $A$ and $B$ are said to  be $s$-order incompatible.
\end{definition}
\begin{remark} This notion is also linear algebraic. The two conditions in the definition are often more practical to be verified by  considering the spanning set of frame vectors.  
\end{remark}
\begin{remark} When $A$ and $B$ are orthonormal bases, we recover the $s$-order incompatibility in Definition \ref{def1}.
\end{remark}

Importantly, it should be noted that if $A$ and $B$ are $s$-order incompatible, then 
\begin{equation} \label{incoms}
    s\ge {\rm spark}(A, B).
\end{equation} This follows from the fact that frames are typically overcomplete. We provide an example.
\begin{ex} \label{ex3} Consider the tight frames $A=\{a_{1}, a_{2}\}$ and $B=\{b_{1}, b_{2}, b_{3}\}$ (see Figure \ref{PQRAFs}),  where
\begin{enumerate}
    \item $a_{1}=(1, 0)$, $a_{2}=(0,1)$;
    \item $b_{1}=\left(0, \frac{\sqrt{2}}{2}\right)$, $b_{2}=\left(\frac{\sqrt{2}}{2}, \frac{1}{2}\right)$ and $b_{3}=\left(\frac{\sqrt{2}}{2}, -\frac{1}{2}\right)$.
\end{enumerate}
\begin{figure*}[h]
	\centering
	\subfloat[(1)]{\includegraphics[width=1.6in]{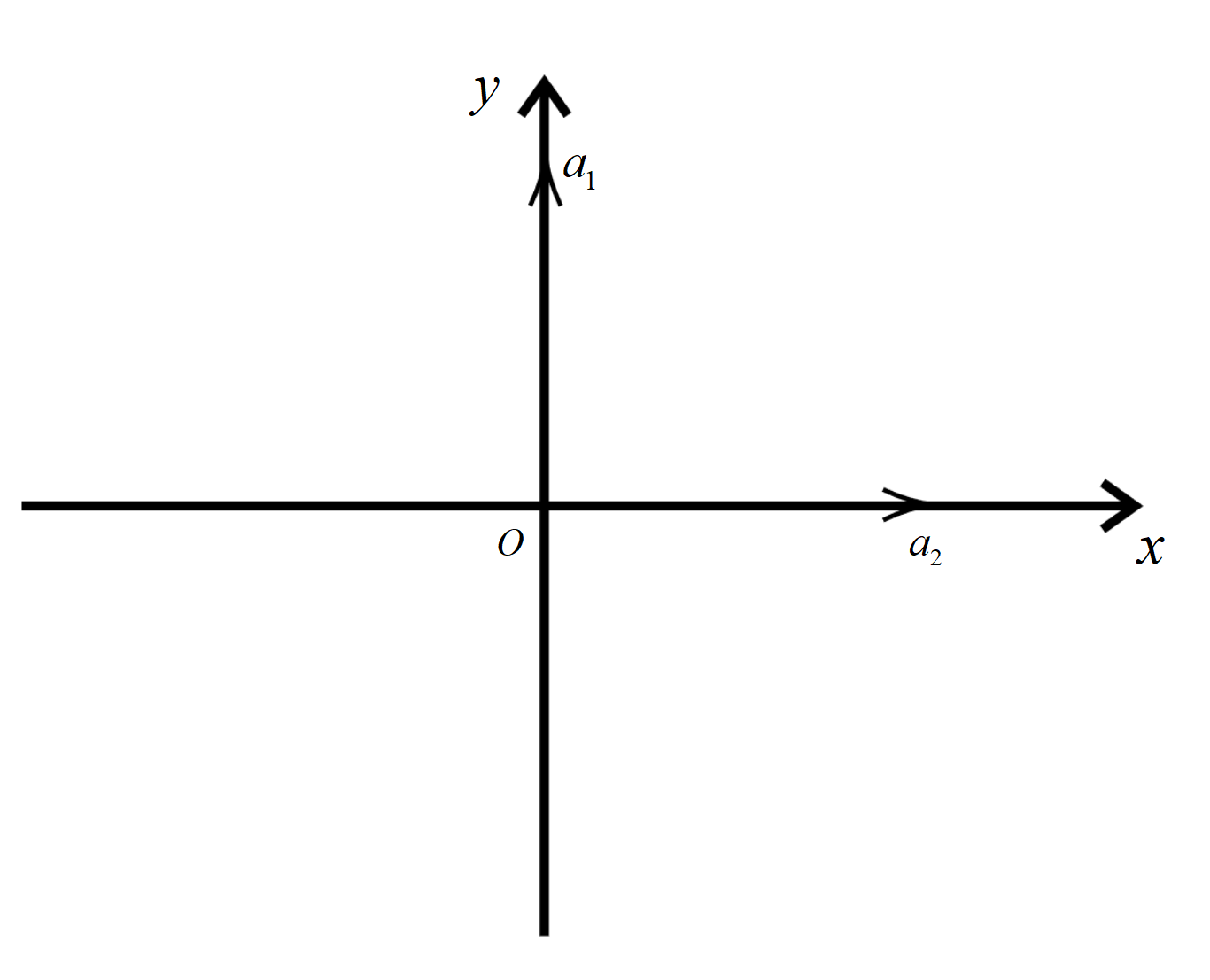}%
		\label{pic_trial_2_0}}\,
	\subfloat[(2)]{\includegraphics[width=1.8in]{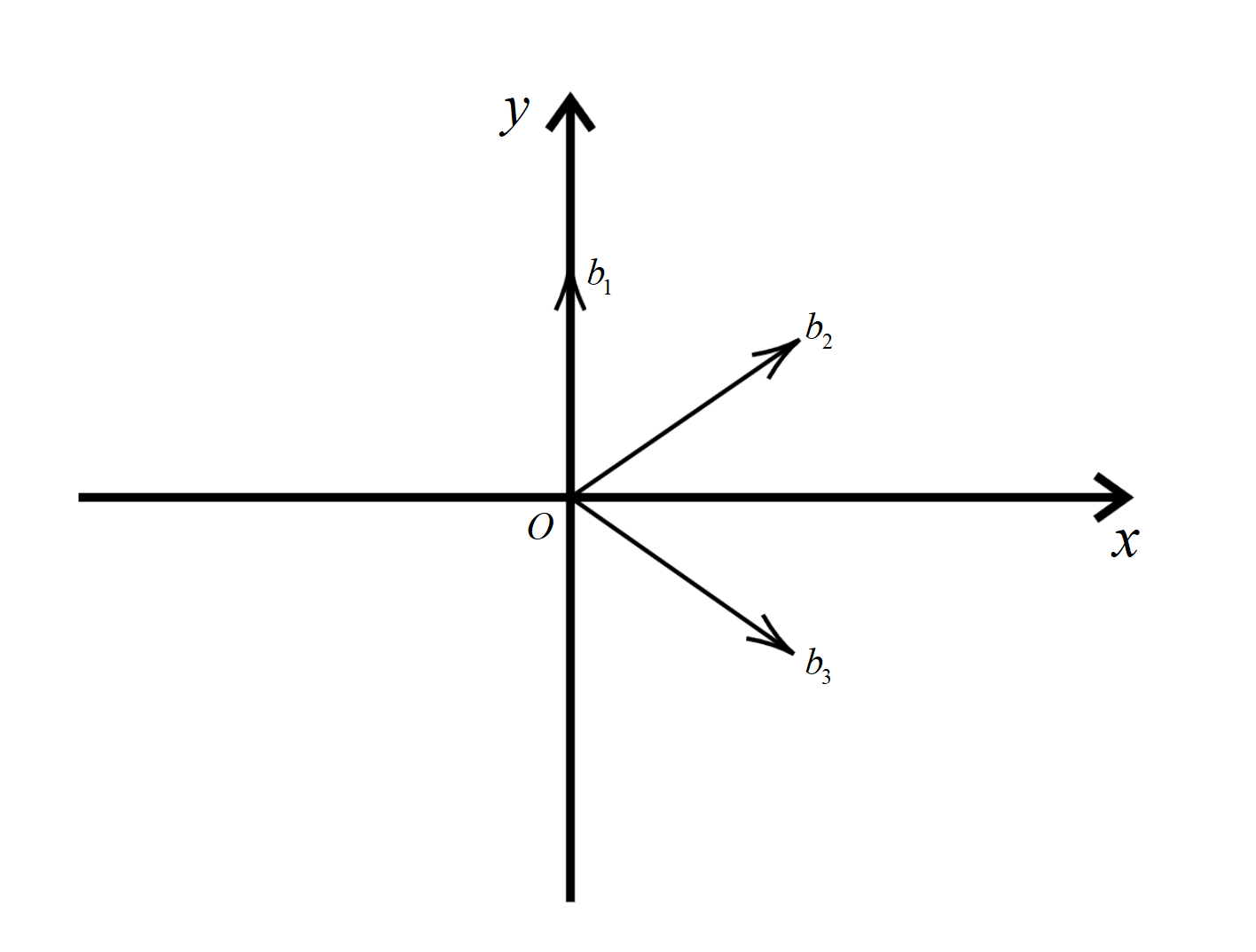}%
		\label{PQRAF_re2}}\,
	\caption{Frames vectors of $A$ and $B$ }
	\label{PQRAFs}
\end{figure*}
It is not hard to verify
\begin{equation*}
    \min(n_{A}(x)+n_{B}(x))=3, \quad {\rm whereas}\quad {\rm spark}(A, B)=2.
\end{equation*}    
\end{ex}

At the end of this section, we show the incompatibility and the support uncertainty relation are closely related as expected.
For a pure state $\ket{\varphi}$, let $n_{A}(\ket{\varphi})$ and $n_{B}(\ket{\varphi})$ denote the number of non-zero elements in the sequences $\{\langle a_{j}, \varphi\rangle\}_{j=1}^{m}$ and  $\{\langle b_{k}, \varphi\rangle\}_{k=1}^{n}$, respectively.

\begin{theorem}[tight support uncertainty relation] \label{sur} Let $A$ and $B$ be two $s$-order incompatible tight frames. Then, for any non-zero pure state $\ket{\varphi}$, we have 
\begin{equation*}
    n_{A}(\ket{\varphi})+n_{B}(\ket{\varphi})\ge s,
\end{equation*} 
where $s$ is optimal.
\end{theorem}
 \begin{proof} Suppose $\min_{\ket{\varphi}}\{ n_{A}(\ket{\varphi})+n_{B}(\ket{\varphi})\}=t$. This means the following two:
 \begin{enumerate}
     \item There exists a pure state $\ket{\varphi}$, such that 
$ n_{A}(\ket{\varphi})+n_{B}(\ket{\varphi})=t$. Namely for such $\ket{\varphi}$, let $S=\{k:\langle a_{k}, \varphi\rangle \neq 0\}$ and $T=\{j: \langle x, b_{j}\rangle \neq 0\}$,  we have $|S|+|T|=t$,  
 \begin{equation*}
        \sum_{k\in S}|\langle  a_{k}, \varphi\rangle|^{2}=\alpha\|x\|^{2} \quad {\rm and} \quad \sum_{j\in T}|\langle  b_{j}, \varphi\rangle|^{2}=\beta\|x\|^{2}
    \end{equation*} 
    according to the definition of tight frames.
    \item There does not exist  a nonzero state $\ket{\phi}$ such that $ n_{A}(\ket{\phi})+n_{B}(\ket{\phi})<t$. It is equivalent to say for any  nonempty sets $S\subset I$ and $T\subset J$ satisfying $|S|+|T|<t$, and for  any non-zero $x\in \mathscr{H}$, at most one of the following holds
    \begin{equation*}
        \sum_{k\in S}|\langle x, a_{k}\rangle|^{2}=\alpha\|x\|^{2}, \quad {\rm and} \quad \sum_{j\in T}|\langle x, b_{j}\rangle|^{2}=\beta\|x\|^{2}.
    \end{equation*} %Otherwise if there exists such  $\ket{\phi}$, by the definition of tight frames, when $k\in S^{c}$, $\langle y, a_{k}\rangle$=0, and when $j\in T^{c}$, 
  %  $\langle y, b_{j}\rangle=0$. Therefore 
   % $ n_{A}(\ket{\phi})+n_{B}(\ket{\phi})<t$.
 \end{enumerate}
Comparing with Definition \ref{ed1}, the above indeed means that $A$ and $B$ are $t$-order incompatible. Therefore $t=s$, the proof is complete.
 \end{proof}

\section{Uncertainty relation from  $s$-order incompatibility} \label{se4}
The main aim of this section is to derive an uncertainty relation  based on  $s$-order incompatibility. While these bounds can be obtained using the support uncertainty relation in Theorem \ref{sur} combined with a standard  compactness argument,  the constant derived through this approach remains unclear (as also noted in \cite{gj}). We study it from a new linear algebra perspective.
 
 Let \( \{a_k\}_{k=1}^{m} \) and \( \{b_j\}_{j=1}^{n} \) be two tight frames in a Hilbert space $\mathscr{H}$. After normalization,   each set corresponds to a rank-one POVM.
\begin{definition}[Minimal Reconstruction Number] For frames $A$ and $B$, define
\[
 t_{\min}= \min\{t \,|\, \forall S\subsetneq I,  T\subsetneq J, |S| + |T| > t,\operatorname{span} \left(\{ a_k \}_{k \in S} \cup \{ b_j \}_{j \in T}\right) =\mathscr{H} \}.
\]   
\end{definition}

%\begin{enumerate}
  %  \item For any $S\subsetneq I$ and $T\subsetneq J$ with $|S| + |T| > t_{\min}$, we have \[{\rm span} \left(\{ a_k \}_{k \in S} \cup  \{ b_j \}_{j \in T} \right)=\mathscr{H}.\]
 %   \item  There exist $\emptyset\neq S_{A}\subset A$ and $\emptyset\neq S_{B}\subset B$, such that $|S_{A}|+|S_{B}|=t_{\min}$ and ${\rm span}(\{S_{A}\}\cup \{S_{B}\})\neq \mathscr{H}$.
%\end{enumerate}
It is shown that the quantities  $s$ and $t_{\min}$ are intrinsically related.

\begin{theorem} \label{th1}
    Suppose two tight frames $A$ and $B$ are \( s \)-order incompatible. Then  
\[
s + t_{\min} = m+n.
\]
\end{theorem}
\begin{proof}
    We  show that \( t_{\min} = m + n - s \).

\textbf{1)} For any subsets $S\subsetneq I$ and $T\subsetneq J$ satisfying  \( |S| + |T| > m+n-s \), it follows that their complements $S^{c}$ and $T^{c}$ satisfy
\[
|S^c| + |T^c| < s.
\]
%where $S^{c}\subset I$ and $T^{c}\subset J$ are non-empty.
Since the frames $\{a_{k}\}_{k=1}^{m}$ and $\{b_{j}\}_{j=1}^{n}$ are \( s \)-order incompatible,  thus by the equivalent Definition \ref{ed1}, for any non-zero vector \( x \in \mathscr{H} \), the following two identities
\[
\sum_{k \in S^c} | \langle x, a_k \rangle |^2 = \alpha \| x \|^2 \quad \mbox{and}\quad \sum_{j \in T^c} | \langle x, b_j \rangle |^2 = \beta \| x \|^2
\]
cannot simultaneously hold. This implies that
\[
\sum_{k \in S^c} | \langle x, a_k \rangle|^2 + \sum_{j \in T^c} | \langle x, b_j \rangle |^2 < (\alpha+\beta)\| x \|^2, \qquad \forall\, x\neq 0.
\]
It is equivalent to say, for any non-zero $x\in \mathscr{H}$,
\[
\sum_{k \in S} | \langle x, a_k \rangle|^2 + \sum_{j \in T} | \langle x, b_j \rangle |^2 >0,
\]
thus \( \left(\operatorname{span} \{ a_k \}_{k \in S^c} \cup  \{ b_j \}_{j \in T^c} \right)^{\perp}= \{0\} \). It follows that \[\operatorname{span} \left(\{ a_k \}_{k \in S^c} \cup  \{ b_j \}_{j \in T^c}\right) = \mathscr{H}.  \]
Therefore, \( t_{\min} \leq m + n - s \).

\textbf{2)} By the definition of $s$-order incompatibility,  there exist two non-empty subsets \( S\subset I \) and \( T\subset J \) satisfying \( |S| + |T| = s \), such that for some nonzero \( x \),  the following hold
\[
\sum_{k \in S} | \langle x, a_k \rangle |^2 = \alpha \| x \|^2, \quad \mbox{and}\quad  \sum_{j \in T} | \langle x, b_j \rangle |^2 = \beta \| x \|^2.
\]
This implies that for any $S^{c}\subsetneq I$ and $T^{c}\subsetneq J$ satisfying $|S^{c}|+|T^{c}|=m+n-s$,  there exists  non-zero $x\in \mathscr{H}$ such that 
\begin{equation*}
    \sum_{k \in S^c} | \langle x, a_k \rangle |^2+\sum_{j \in T^{c}} | \langle x, b_j \rangle |^2=0.
\end{equation*}
Thus \( \left({\rm span} \left(\{ a_k \}_{k \in S^c} \cup  \{ b_j \}_{j \in T^{c}}\right)\right)^{\perp} \neq \{0\} \), i.e.,
 \[ {\rm span} \left(\{ a_k \}_{k \in S^c} \cup  \{ b_j \}_{j \in T^c}\right) \neq \mathscr{H}. \]
Hence \( t_{\min} \ge m + n - s \).

Collecting all, we obtain \( t_{\min} = m + n - s \), completing the proof.
\end{proof}
\begin{remark}
    The quantity $t_{\min}$  provides a lower bound (i.e., $t_{\min}+1$) on the number of measurements required to reconstruct a quantum state from two rank-one POVMs.
\end{remark}

 By Theorem \ref{th1}, for any non-empty subsets $S\subset I $ and \( T \subset J \) satisfying \( |S| + |T| < s \), we have
\[
\operatorname{span} \left(\{ a_k \}_{k \in S^c} \cup  \{ b_j \}_{j \in T^c}\right) = \mathscr{H}. 
\]
Thus the set $\{a_{k}\}_{k\in S^{c}} \cup \{b_{j}\}_{j\in T^{c}}$ forms a frame in $\mathscr{H}$ (by \cite[Lemma 1.2]{ckp}), whose lower frame bound is denoted by  $C_{S, T}$. Furthermore, denoting $C_{s}=\min_{|S| + |T| < s} C_{S, T}$. The following uncertainty relation follows.

\begin{theorem}
 Let \( \{a_k, 1\le k\le m\} \) and \( \{b_j, 1\le j\le n\} \)  be two \( s \)-order incompatible tight frames in $\mathscr{H}$. Then there exists a constant \( C \), such that for all subsets \( S \subseteq I \) and \( T \subseteq J \) satisfying \( |S| + |T| < s \), we have 
\begin{equation} \label{nc1}
    \| x \|^2 \leq C \left( \sum_{k \in S^c} | \langle x, a_k \rangle |^2 + \sum_{j \in T^c} | \langle x, b_j \rangle |^2 \right).
\end{equation}
where   $C=1/\min\{\alpha, \beta, C_{s}\}$, in which $\alpha$ and $\beta$ are the frame constants.
\end{theorem}

\begin{remark} If the sets $S$ and $T$ are fixed, the constant in \eqref{nc1} can be chosen as $C=1/\min\{\alpha, \beta, C_{S, T}\}$. 
\end{remark}

 For a matrix $M$, the dictionary coherence $\mu(A)$ is defined as the maximum correlation between any two distinct columns \begin{equation}\label{qer}
    \mu(M)=\max_{m\neq n}|\langle a_{m}, a_{n}\rangle|.
\end{equation}
If the columns of  $M$ are normalized to unit norm,  a lower bound of its spark is given in terms of its dictionary coherence by 
\begin{equation*}
    {\rm spark}(M)\ge 1+\frac{1}{\mu(M)}.
\end{equation*}
Thus, the uncertainty relation \eqref{nc1} is closely related to  the  Ghobber-Jaming inequality \cite{gj}. However, by inequality \eqref{incoms}, it is seen that \eqref{nc1} is  more general, while the constant provided in \cite{gj} is  explicit  particularly for orthogonal bases.

\section{Ghobber-Jaming uncertainty  for frames}\label{se5}
This section is dedicated to generalizing the Ghobber-Jaming uncertainty inequality in \cite{gj} for two frames $A=\{a_{k}\}_{k=1}^m$ and $B=\{b_{j}\}_{j=1}^{n}$ (not necessarily tight) in terms of their coherence, which is given by 
\begin{equation*}
M(A,B)=\max_{k,j}|\langle a_{k}, b_{j}\rangle| 
\end{equation*}
with explicit constant. This question was raised at the end of \cite{kk}.
\begin{theorem}
    Let $A$ and $B$ be two frames in \( \mathscr{H} \), with lower and upper frame bounds \( \alpha_1, \beta_1 \) for $A$ and \( \alpha_2, \beta_2 \) for $B$, respectively. Suppose that for any subsets \( S \subseteq \{1, \dots, m\} \) and \( T \subseteq \{1, \dots, n\} \), satisfying 
    \begin{equation*}|S||T|<\frac{\beta_{1}}{\alpha_{2}}\cdot\frac{1}{M(A^{*}, B)^{2}}.   
    \end{equation*}
Then, for any \( x \in \mathscr{H} \), we have 
\[
C(S, T)\| x \| \leq \left( \sum_{k \in S^c} | \langle x, a_k \rangle |^2\right)^{\frac{1}{2}} + \left(\sum_{j \in T^c} | \langle x, b_j \rangle |^2 \right)^{\frac{1}{2}},
\]
where 
\begin{equation*}
\begin{split}
    C(S,T)=&\left(1-\left(\frac{\alpha_{2}}{\beta_{1}}\right)^{1/2}M(A^{*}, B)|S|^{\frac{1}{2}}|T|^{\frac{1}{2}}\right)\\
    & \times \left(\max\left\{\beta_{1}^{-1/2},  \left(1+ \left(\frac{\beta_{2}}{\beta_{1}}\right)^{\frac{1}{2}}\right)\alpha_{1}^{-\frac{1}{2}}mM(A^{*}, A)      \right\} \right)^{-1}.
\end{split} 
\end{equation*}
\end{theorem}
\begin{proof}
    Let \( x = x_1 + x_2 \), where  
\[
x_1 = \sum_{k \in S} \langle x, a_k \rangle a_k^*, \quad {\rm and}\quad x_2 = \sum_{k \in S^c} \langle x, a_k \rangle a_k^*,
\]
in which $a_{k}^{*}$ is the canonical dual frame. For $x_{1}$, we have
\begin{equation*}
    \begin{split}
        \sum_{j\in T}|\langle x_{1}, b_{j}\rangle|^{2}&=\sum_{j\in T}\left|\sum_{k\in S}\langle x, a_{k}\rangle \langle a_{k}^{*}, b_{j}\rangle\right|^{2}\\
      &\le \sum_{j\in T}\left(\sum_{k\in S}|\langle x, a_{k}\rangle|^{2}\right)\left(\sum_{k\in S}\left|\langle a_{k}^{*}, b_{j}\rangle\right|^{2}\right)\\
       &\le \alpha_{2} M(A^{*}, B)^{2} |S||T|\|x\|^{2}.
    \end{split}
\end{equation*}
Therefore,
\begin{equation*}
\begin{split}
    \sum_{j\in T^{c}}|\langle x_{1}, b_{j}\rangle|^{2}&=\sum_{j=1}^{n}\left|\langle x_{1}, b_{j}\rangle\right|^{2}-\sum_{j\in T}|\langle x_{1}, b_{j}\rangle|^{2}\\&\ge \beta_{1}\|x_{1}\|^{2}
-\alpha_{2}M(A^{*}, B)^{2}|S||T|\|x\|^{2}.\end{split}
\end{equation*}
It yields that
\begin{equation}\label{231}
    \begin{split}
\beta_{1}^{\frac{1}{2}}\|x_{1}\| \le &\left( \sum_{j\in T^{c}}|\langle x_{1}, b_{j}\rangle|^{2}+\alpha_{2}M(A^{*}, B)^{2}|S||T|\|x\|^{2} \right)^{\frac{1}{2}}\\
 \le& \left(\sum_{j\in T^{c}}|\langle x_{1}, b_{j}\rangle|^{2}\right)^{\frac{1}{2}}+\alpha_{2}^{\frac{1}{2}}M(A^{*}, B)|S|^{\frac{1}{2}}|T|^{\frac{1}{2}}\|x\|\\
\le &\left(\sum_{j\in T^{c}}|\langle x, b_{j}\rangle|^{2}\right)^{\frac{1}{2}}+ \left(\sum_{j\in T^{c}}|\langle x_{2}, b_{j}\rangle|^{2}\right)^{\frac{1}{2}}\\
&+  \alpha_{2}^{\frac{1}{2}}M(A^{*}, B)|S|^{\frac{1}{2}}|T|^{\frac{1}{2}}\|x\|\\
\le & \left(\sum_{j\in T^{c}}|\langle x, b_{j}\rangle|^{2}\right)^{\frac{1}{2}} +\beta_{2}^{\frac{1}{2}}\|x_{2}\|+ \alpha_{2}^{\frac{1}{2}}M(A^{*}, B)|S|^{\frac{1}{2}}|T|^{\frac{1}{2}}\|x\|.
    \end{split}
\end{equation}
On the other hand, 
\begin{equation}
\label{232}
    \begin{split}
\alpha_{1}\|x_{2}\|^{2}&\le \sum_{k=1}^{m}|\langle x_{2}, a_{k}\rangle|^{2}\le \sum_{k=1}^{m} \left( \sum_{i\in S^{c}}|\langle x, a_{i}\rangle|^{2}\right)
\left( \sum_{i\in S^{c}}|\langle a_{i}^{*}, a_{k}\rangle|^{2} \right)\\
&\le m^{2}M(A^{*}, A)^{2} \sum_{k\in S^{c}}|\langle x, a_{k}\rangle|^{2}.
   \end{split}
\end{equation}
By \eqref{231} and \eqref{232}, we have 
\begin{equation*}
    \begin{split}
\|x\|\le \|x_{1}\|+\|x_{2}\| \le & \frac{1}{\beta_{1}^{\frac{1}{2}}} \left( \sum_{j\in T^{c}}|\langle x, b_{j}\rangle|^{2}\right)^{\frac{1}{2}}
+\left( 1+\left(\frac{\beta_{2}}{\beta_{1}}\right)^{1/2}   \right)\|x_{2}\|\\&+\left(\frac{\alpha_{2}}{\beta_{1}}\right)^{1/2}M(A^{*}, B)|S|^{\frac{1}{2}}|T|^{\frac{1}{2}}\|x\|. 
   \end{split}
\end{equation*}
Regrouping terms, we obtain
\begin{equation*}
    \begin{split}
C(S, T)\|x\|\le \left( \sum_{k\in S^{c}}|\langle x, a_{k}\rangle|^{2}  \right)^{\frac{1}{2}} +\left(\sum_{j\in T^{c}}|\langle x, b_{j}\rangle|^{2}\right)^{1/2},
   \end{split}
\end{equation*}
where 
\begin{equation*}
\begin{split}
    C(S,T)=&\left(1-\left(\frac{\alpha_{2}}{\beta_{1}}\right)^{1/2}M(A^{*}, B)|S|^{\frac{1}{2}}|T|^{\frac{1}{2}}\right)\\
    & \times \left(\max\left\{\beta_{1}^{-1/2},  \left(1+ \left(\frac{\beta_{2}}{\beta_{1}}\right)^{\frac{1}{2}}\right)\alpha_{1}^{-\frac{1}{2}}mM(A^{*}, A)      \right\} \right)^{-1}.
\end{split} 
\end{equation*}
\end{proof}
\begin{remark}
    It is interesting to further improve the bound in terms of 
    \begin{equation*}
        \mu_{r}(A,B)=\max_{1\le j\le n}\left( \sum_{k=1}^{m}|\langle a_{k}, b_{j}\rangle|^{r'}\right)^{\frac{r}{r'}},
    \end{equation*}
    where $r\in [1, 2]$ and $1/r+1/r'=1$, similarly as done in \cite{RT}.
\end{remark}
\section{Incompatibility of multiple tight frames}
The main goal of this section is to extend the notion of
$s$-order incompatibility for  multiple tight frames ($n\ge 3$). When dealing with multiple frames, a suitable approach should consider  the overall coherence of the entire collection, rather than analyzing each pair separately and subsequently combining the results. To illustrate the differences, we conclude this section with an example.

\begin{definition} \label{fd1} Suppose $A_{i}=\{a_k^{(i)}, 1\le k\le m_{i}\}$ are $n$ tight frames. They  are said to be  $s$-order incompatible if:
\begin{enumerate}
    \item  For any non-empty subsets $S_{i}\subset I_{i}$ satisfying $\sum_{i=1}^{n}|S_{i}|<s$, and for any non-zero $x\in \mathscr{H}$, the following 
    \begin{equation*}
        \sum_{k\in S_{i}}\left|\left\langle x, a_{k}^{(i)}\right\rangle\right|^{2}=\alpha_{i}\|x\|^{2}, 
    \end{equation*} cannot hold for all $i=1,2,\ldots, n$.
    \item There exist non-empty subsets $S_{i}\subseteq I_{i}$ with $\sum_{i=1}^{n}|S_{i}|=s$, and a non-zero $x\in \mathscr{H}$, such that 
    \begin{equation*}
        \sum_{k\in S_{i}}\left|\left\langle x, a_{k}^{(i)}\right\rangle\right|^{2}=\alpha_{i}\|x\|^{2}
    \end{equation*}
   holds for  all $i=1,2,\ldots, n$.
\end{enumerate}
\end{definition}
\begin{theorem} The tight frames $A_{1}, A_{2}, \cdots, A_{n}$ are  $s$-order incompatible if and only if  for any non-zero pure state $\ket{\varphi}$, it holds that
\begin{equation*}
    \sum_{k=1}^{n}n_{A_{k}}(\ket{\varphi})\ge s.
\end{equation*}
The constant  $s$ here is optimal.
\end{theorem}

\begin{proof} Denote \[n_{\min}=\min_{\ket{\varphi}\neq 0} \sum_{k=1}^{n}n_{A_{k}}(\ket{\varphi})=t.\] Then  there exists a pure state \( \ket{\psi}\) such that
\[
\sum_{k=1}^{n}n_{A_{k}}(\ket{\psi}) = t,
\]
and there does not exist a pure state \( \ket{\psi'} \) such that \( \sum_{k=1}^{n}n_{A_{k}}(\ket{\psi'})  < t \).

For such \( \ket{\psi} \), there exist nonempty subsets \( S_{A_{k}} \subseteq A_{k} \) such that \( |S_{A_{k}}| = n_{A_{k}}(\ket{\psi}) \),  and \( \ket{\psi} \in \bigcap_{k=1}^{n}{\rm span}(S_{A_{k}})\). The nonexistence of such \( \ket{\psi'} \) implies that there do not exist nonempty subsets \( S_{A_k} \subseteq A_{k} \)  such that \( |S_{A_{k}}| = n_{A_{k}}(\ket{\psi'}) \),  and \( \ket{\psi'} \in \bigcap_{k=1}^{n}{\rm span}(S_{A_{k}})\). These two conditions precisely coincide with the two conditions in Definition \ref{fd1}, which implies that $t=s$. Then the claim follows. \( \square \)
    
\end{proof}

We show that the multiple cases ($n\ge 3$) can not be derived from the classical two-frame case. 

\begin{theorem} Suppose tight frames $A_{1}, A_{2}, \cdots, A_{n}$ are $s$-order incompatible and each pair $A_{i}$ and $A_{j}$ are $s_{ij}$-order incompatible, then we have 
    \begin{equation}\label{per}
    \frac{1}{2}\sum_{i<j}s_{ij}\le s.
\end{equation}
\end{theorem}
\begin{proof} Since $A_{1}, A_{2}, \cdots, A_{n}$ are $s$-order incompatible, then there exists a nonzero $x\in \mathscr{H}$ such that 
\begin{equation} \label{p1}
    n_{A_{1}}(x)+n_{A_{2}}(x)+\ldots+n_{A_{n}}(x)\ge s.
\end{equation}
On the other side, since $A_{i}$ and $A_{j}$ are $s_{ij}$-order incompatible, we have 
\begin{equation}\label{p2}
    n_{A_{i}}(x)+n_{A_{j}}(x)\ge s_{ij}.
\end{equation}
Combining \eqref{p1} and \eqref{p2}, we get the result.
    
\end{proof}

The following example shows that the strict inequality in \eqref{per} can occur.

\begin{ex} \label{expli} Consider the frames $A=\{a_{1}, a_{2}, a_{3}\}$, $B=\{b_{1}, b_{2}, \ldots, b_{5}\}$ and $C=\{c_{1}, c_{2}, c_{3}, c_{4}\}$, where
\begin{enumerate}
    \item $
    a_{1}=(2\sqrt{3}, 1), a_{2}=(-2, -\sqrt{3}), a_{3}=(0, 2\sqrt{3});
$
\item $b_{1}=(1, 1), b_{2}=(\frac{1}{2}, \frac{1}{2}), 
b_{3}=(0, 1), b_{4}=(2, -\sqrt{3}),$\\ $ b_{5}=(-\left(2\sqrt{3}-\frac{5}{4}\right)^{\frac{1}{2}}, -\left(2\sqrt{3}-\frac{5}{4})^{\frac{1}{2}}\right)$;
\item  $c_{1}=(2, -\sqrt{3}), c_{2}=(2\cdot 3^{\frac{1}{4}}, 2\cdot 3^{\frac{1}{4}}), c_{3}=(-2, \sqrt{3}), c_{4}=(0, -\sqrt{2})$.
\end{enumerate}
\begin{figure*}[h]
	\centering
	\subfloat[(1)]{\includegraphics[width=1.7in]{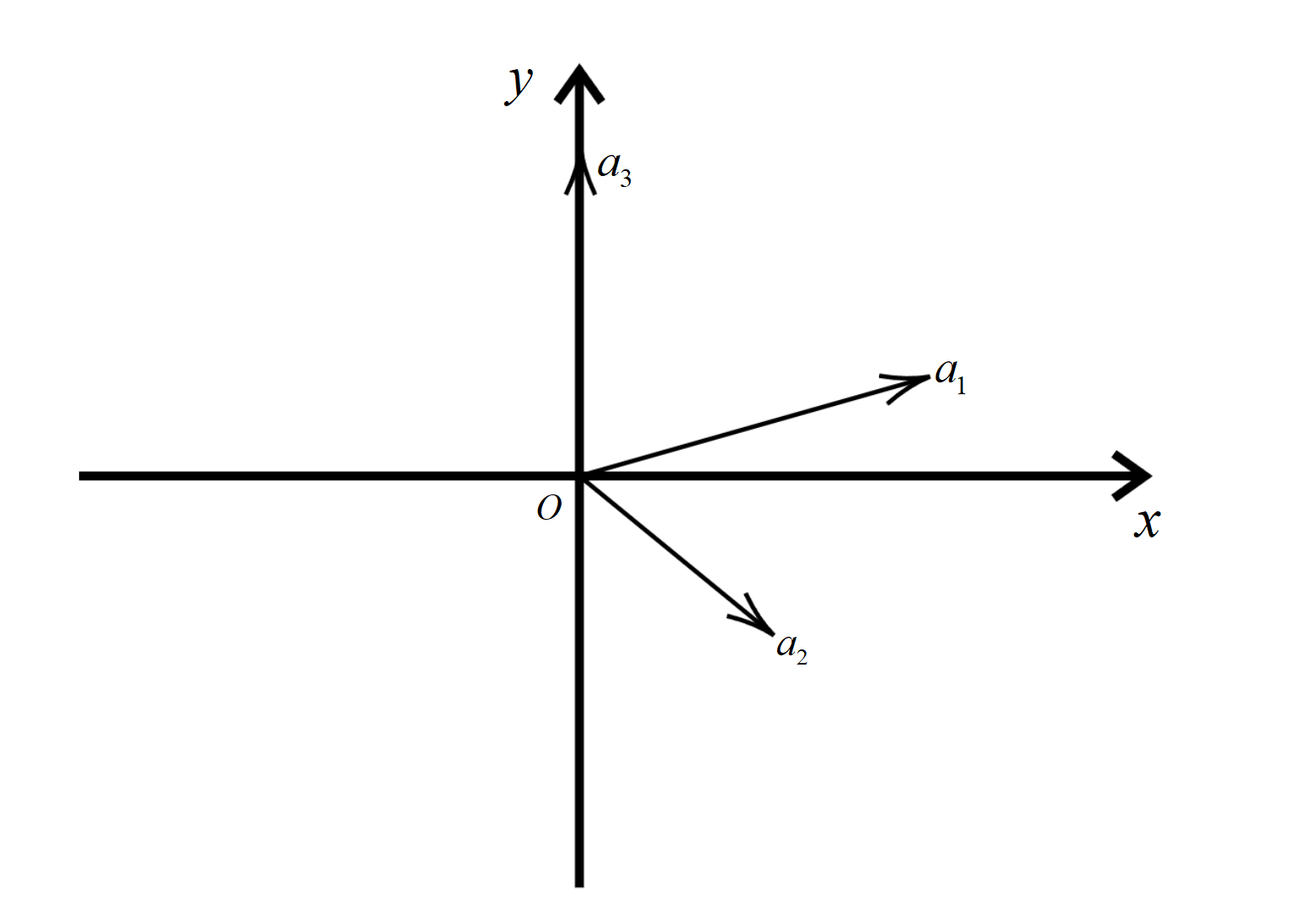}}%
		%\label{pic_trial_2_0}}\,
	\subfloat[(2)]{\includegraphics[width=1.5in]{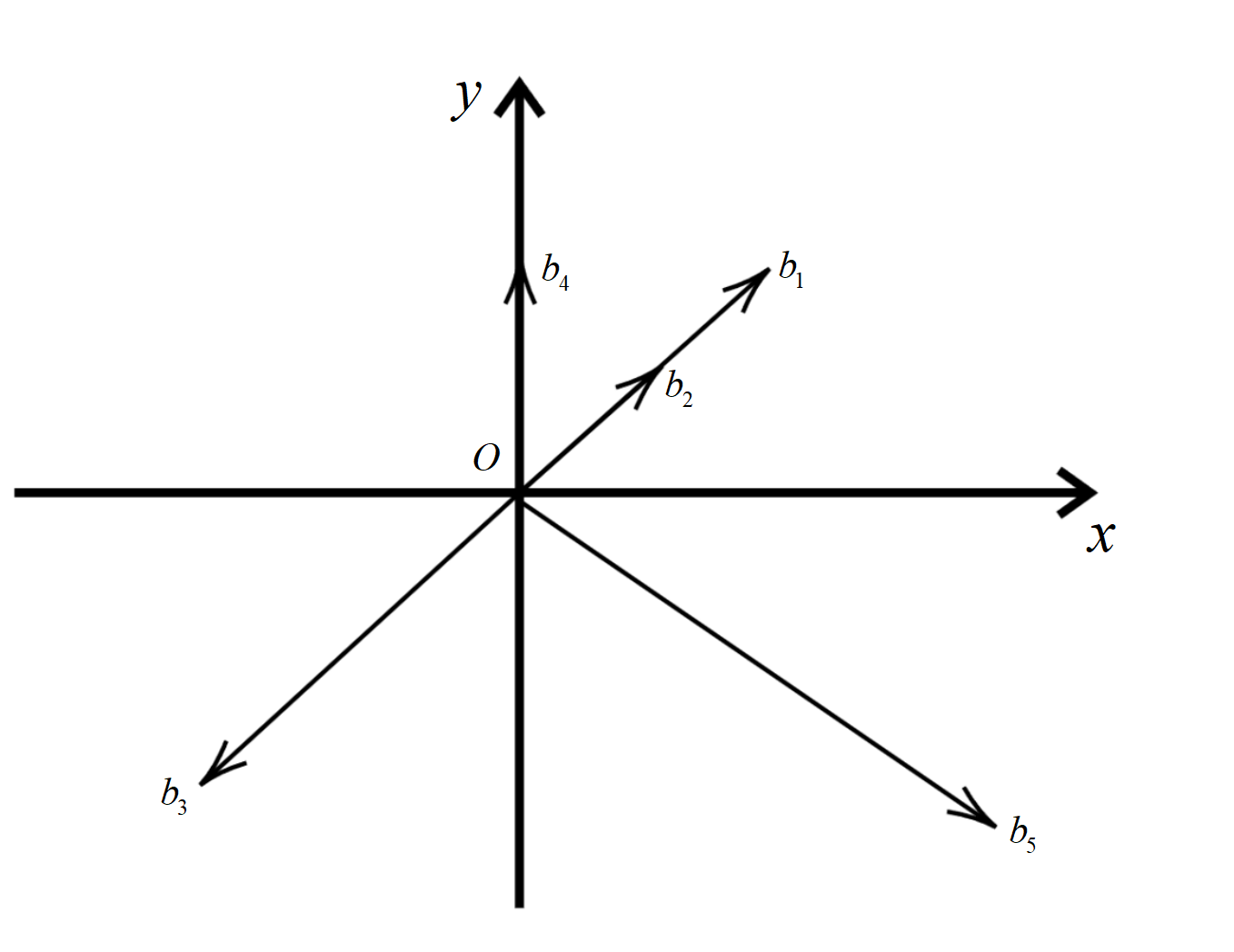}}%
		%\label{PQRAF_re2}}\,
        \subfloat[(3)]{\includegraphics[width=1.5in]{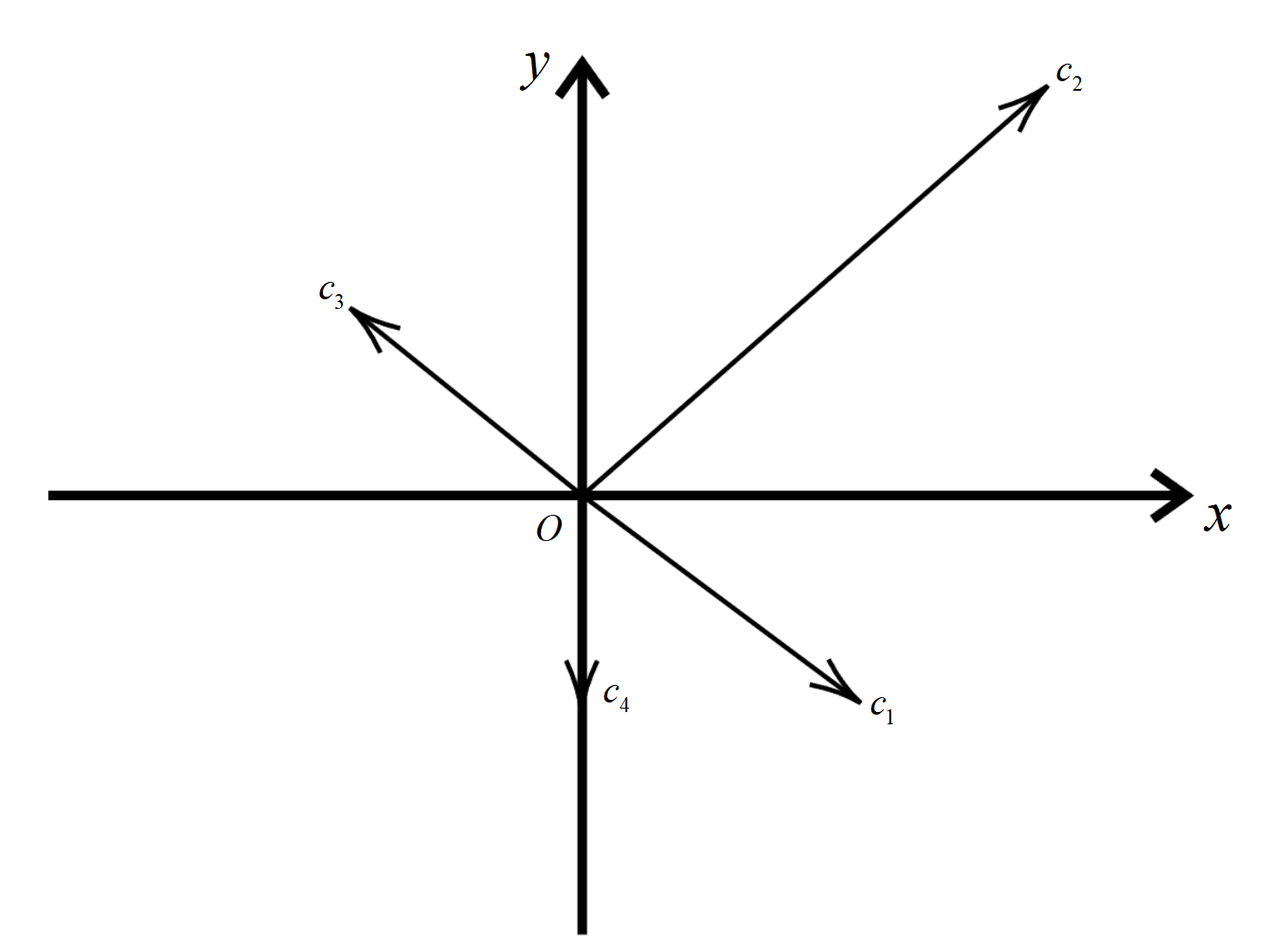}}%
	\caption{frames vectors of $A, B$ and $C$}
	\label{PQRAFs1}
\end{figure*}
Straightforward computations show that $s_{AB}=5$, $s_{BC}=5$, $s_{AC}=4$ and $s=8$. Therefore,
\begin{equation*}
    \frac{s_{AB}+s_{BC}+s_{AC}}{2}<s.
\end{equation*}
 This means that the cases for multiple measurements $n\ge 3$  differs from the uncertainty relation for two bases.
\end{ex}

\section{Conclusion}
In this paper, we show that the $s$-order incompatibility of two orthonormal bases coincides with the spark of matrix formed by the measurement vectors, thereby linking quantum measurement with compressed sensing. The main contribution of this work is the classification of the incompatibility of two and multiple ($n\ge 3$) tight frames (rank-one POVMs), leading to sharp support uncertainty relations and new generalizations of Ghobber-Jaming uncertainty relation. Potential directions 
for future research include applications in quantum tomography, cryptographic protocol design, and explicit calculations of the order of incompatibility.

\section*{Acknowledgments}
This work  was partially supported  by NSFC Grant No.12101451.

\bibliographystyle{amsplain}

\end{document}